\documentclass[12pt]{article}
\usepackage{caption}
\usepackage{diagbox}
\usepackage{verbatim}
\usepackage[usenames]{color}
\usepackage{amssymb}
\usepackage{graphicx}
\usepackage{amscd}
\usepackage{algorithm}
\usepackage{enumerate}
\usepackage[noend]{algpseudocode}

\usepackage{tikz}
\usetikzlibrary{decorations.text}

\usepackage{listings}

\providecommand{\keywords}[1]
{
  \small	
  \textbf{\textit{Keywords---}} #1
}

\definecolor{verbgray}{gray}{0.9}

\lstnewenvironment{code}{%
  \lstset{backgroundcolor=\color{verbgray},
 frame=single,
  language=C,
  framerule=0pt,
  basicstyle=\ttfamily,
  showstringspaces=false,
  commentstyle=\color{blue}\textit,
  keywordstyle=\color{black}\bf,
  numbers=none,
  keepspaces=true,
  columns=fullflexible}}{}

\usepackage[colorlinks=true,
linkcolor=webgreen,
filecolor=webbrown,
citecolor=webgreen]{hyperref}

\definecolor{webgreen}{rgb}{0,.5,0}
\definecolor{webbrown}{rgb}{.6,0,0}

\usepackage{color}
\usepackage{fullpage}
\usepackage{float}

\usepackage{graphics,amsmath,amssymb}
\usepackage{amsthm}
\usepackage{amsfonts}
\usepackage{latexsym}
\usepackage{epsf}

\DeclareMathOperator{\rankB}{rankB}
\DeclareMathOperator{\rankU}{rankU}

\DeclareMathOperator{\LPS}{PBA}
\newsavebox{\mybox}

\usepackage{breakurl}

\begin{document}

\theoremstyle{plain}
\newtheorem{theorem}{Theorem}
\newtheorem{corollary}[theorem]{Corollary}
\newtheorem{lemma}[theorem]{Lemma}
\newtheorem{proposition}[theorem]{Proposition}

\theoremstyle{definition}
\newtheorem{definition}[theorem]{Definition}
\newtheorem{example}[theorem]{Example}
\newtheorem{conjecture}[theorem]{Conjecture}

\theoremstyle{remark}
\newtheorem{remark}[theorem]{Remark}

\author{Daniel Gabric \\
Department of Math/Stats \\
University of Winnipeg \\
Winnipeg, MB  R3B 2E9 \\
Canada\\
\href{mailto:d.gabric@uwinnipeg.ca}{\tt d.gabric@uwinnipeg.ca} }
\date{}

\title{Ranking and unranking bordered and unbordered words}

\maketitle

\begin{abstract}
A \emph{border} of a word $w$ is a word that is both a non-empty proper prefix and suffix of $w$. If $w$ has a border, then it is said to be \emph{bordered}; otherwise, it is said to be \emph{unbordered}. The main results of this paper are the first algorithms to rank and unrank length-$n$ bordered and unbordered words over a $k$-letter alphabet. We show that, under the unit-cost RAM model, ranking bordered and unbordered words can be done in $O(kn^3)$ time using $O(n)$ space, and unranking them can be done in $O(n^4k\log k)$ time using $O(n)$ space. 
\end{abstract}


\section{Introduction}

A word $u$ is said to be a \emph{border} of a word $w$ if $u$ is a non-empty proper prefix and suffix of $w$. If $w$ has a border, then it is said to be \emph{bordered}; otherwise, it is said to be \emph{unbordered}. For example, the word ${\tt alfalfa}$ has ${\tt a}$ and ${\tt alfa}$ as borders, so it is bordered. The word ${\tt unbordered}$ has no borders, so it is unbordered. Bordered and unbordered words are important combinatorial objects in computer science. Borders are fundamental to the most efficient string searching algorithms~\cite{Knuth&Morris&Pratt:1977, Boyer&Moore:1977}. Borders are also useful for figuring out statistics on words~\cite{Guibas&Odlyzko:1980,Guibas&Odlyzko:1981}. Additionally, unbordered words appear as optimal sync words in frame synchronization~\cite{Massey:1972, Scholtz:1980}.

In computer science, three problems often arise when a particular combinatorial object is discovered. We want to know how to efficiently count, exhaustively generate, and randomly generate distinct instances of the object. Nielsen~\cite{Nielsen:1973} gave a recurrence to count the number $u_k(n)$ of length-$n$ unbordered words over a $k$-letter alphabet. Since every word is either bordered or unbordered, the number of length-$n$ bordered words over a $k$-letter alphabet is $k^n-u_k(n)$. Using his recurrence he showed that there are $\Theta(k^n)$ length-$n$ bordered and unbordered words. Nielsen also showed how to exhaustively generate all bordered and unbordered words using a recursive procedure. This procedure requires $O(n)$ time per length-$n$ word generated. The problem of efficiently randomly generating bordered and unbordered words is open.

One way of randomly generating an instance of a combinatorial object is with the use of ranking and unranking algorithms. Let $o_1,o_2,\ldots, o_m$ be an ordered list of distinct combinatorial objects. The \emph{rank} of $o_i$ is $i$, the position of $o_i$ within the list. A ranking algorithm for the list computes $i$ given $o_i$. An unranking algorithm for the list computes $o_i$ given $i$. From here it is easy to see that one can randomly generate one of the objects in the list by randomly generating a rank between $1$ and $m$ and then unranking the object at that rank. A na\"{i}ve way to rank and unrank is to generate the list of objects and then determine the position of a particular object. However, this is not ideal, especially if the number of distinct objects is exponential. Thus, we want efficient ranking and unranking algorithms. By efficient we mean that the ranking and unranking algorithms run in $o(f(n))$ time using $o(f(n))$ space where $f(n)$ is the number of distinct objects of ``order" $n$.

Efficient ranking and unranking algorithms have been discovered for many different combinatorial objects, such as permutations~\cite{Myrvold&Ruskey:2001, Mares&Straka:2007}, trees~\cite{Pallo:1986, Li:1986, Gupta:1983, Gupta:1982}, and necklace variations~\cite{Hartman&Sawada:2019, Sawada&Williams:2017, Adamson:2021,Adamson:2022}. See~\cite{Ruskey:2003} for a more general survey on the topic of ranking and unranking algorithms. In this paper, we consider bordered and unbordered words. We present the first efficient ranking and unranking algorithms for bordered and unbordered words in lexicographic order. 

The rest of the paper is structured as follows. In Section~\ref{section:prel} we introduce some definitions and results that are necessary to prove our main results. In Section~\ref{section:recurrence} we present a recurrence that will serve as the basis for our ranking algorithms. In Section~\ref{section:ranking} we present our ranking algorithms for bordered and unbordered words. In Section~\ref{section:unrank} we show how to unrank bordered and unbordered words. 
In Section~\ref{section:conclusion} we give some concluding remarks and an open problem. For all the algorithms presented in this paper, the unit-cost RAM model is assumed. Under this computational model, integer variables use constant space and integer arithmetic takes constant time. A complete C implementation of our ranking and unranking algorithms can be found here: \url{https://github.com/DanielGabric/RankUnrank}.

\section{Preliminaries}\label{section:prel}
Let $\Sigma_k$ denote the alphabet $\{1,2,\ldots, k\}$ where $1<2<\cdots < k$. Let $\Sigma_k^n$ denote the set of all length-$n$ words over $\Sigma_k$. Let $x=x_1x_2\cdots x_m \in \Sigma_k^m$ and $y=y_1y_2\cdots y_n \in \Sigma_k^n$. Then $x<y$ in \emph{lexicographic order} either if $x$ is a prefix of $u$ or if $x_i< y_i$ for the smallest $i$ such that $x_i\neq y_i$.

Let $w=w_1w_2\cdots w_n$ be a length-$n$ word. The \emph{unbordered prefix indicator} $a[1..n]$ of $w$ is a length-$n$ integer array such that $a[i]=1$ if $w_1w_2\cdots w_i$ is unbordered and $a[i]=0$ otherwise. The \emph{border indicator} $b[1..n]$ of $w$ is a length-$n$ integer array such that $b[i]=1$ if $w$ has a border of length $i$ and $b[i]=0$ otherwise. For example, consider the binary word $w=011101110$. The unbordered prefix indicator of $w$ is $a[1..9] = [1,1,1,1,0,0,0,0,0]$. The border indicator of $w$ is $b[1..9] = [1,0,0,0,1,0,0,0,0]$.

\begin{lemma}\label{lemma:unbPref}
 Let $n\geq 1$ be an integer. Let $w=w_1w_2\cdots w_n$ be a word of length $n$. Then the unbordered prefix indicator of $w$ can be computed in $O(n)$ time.
\end{lemma}
\begin{proof}
Recall the failure function of a word, also known as the longest prefix suffix array, from the Knuth-Morris-Pratt string searching algorithm~\cite{Knuth&Morris&Pratt:1977}. The failure function of $w$ is a length-$n$ integer array $A[1..n]$ such that $A[i]=j$ if and only if the longest proper prefix of $w_1w_2\cdots w_i$ that matches a suffix of $w_1w_2\cdots w_i$ is of length $j$. Thus $A[i]=0$ if and only if $w_1w_2\cdots w_i$ is unbordered. Since the failure function encodes the longest borders of all prefixes of a word, we call it the \emph{Prefix Border Array} ($\LPS$). Let $\LPS[1..n]$ be the $\LPS$ array of the word $w$. One can easily obtain the unbordered prefix indicator of $w$ from $\LPS[1..n]$ since $\LPS[i]=0$ implies $w_1w_2\cdots w_i$ is unbordered. Knuth, Morris, and Pratt~\cite{Knuth&Morris&Pratt:1977} showed that $\LPS[1..n]$ can be computed in $O(n)$ time. Therefore, the unbordered prefix indicator of $w$ can also be computed in $O(n)$ time.
\end{proof}
\begin{lemma}\label{lemma:bordInd}
 Let $n\geq 1$ be an integer. Let $w$ be a word of length $n$. Then the border indicator of $w$ can be computed in $O(n)$ time.
\end{lemma}
\begin{proof}
  A length-$n$ word $w=w_1w_2\cdots w_{n}$ is said to have a \emph{period} $p$ if $w_{i}=w_{i+p}$ for all $1\leq i \leq n-p$. The \emph{auto-correlation}~\cite{Guibas&Odlyzko:1980, Guibas&Odlyzko:1981} of $w$ is a binary word $a_1a_2\cdots a_n$ such that $a_i=1$ if and only if $i$ is a period $w$. It is easy to see that a length-$n$ word has a period $p$ if and only if it has a border of length $n-p$. Thus one can easily obtain the border indicator of $w$ from the auto-correlation of $w$. Corollary 2 in~\cite{Lind:1989} shows how to compute the auto-correlation of a length-$n$ word in $O(n)$ time using the failure function from the Knuth-Morris-Pratt string searching algorithm. Therefore, we can also compute the border indicator of $w$ in $O(n)$ time.
\end{proof}

The following two lemmas will serve as the basis for the recurrence in Section~\ref{section:recurrence}.

\begin{lemma}\label{lemma:shortestUnb}
Let $w$ be a length-$n$ bordered word. Let $u$ be a border of $w$. Then $u$ is the shortest border of $w$ if and only if $u$ is unbordered.
\end{lemma}
\begin{proof} 
We prove the contrapositive of both directions. If $u$ is bordered, then the border of $u$ is a shorter border of $w$. If $u$ is not the shortest border, then there exists a shorter border $u'$ of $w$. But this border is now both a non-empty proper prefix and suffix of $u$, so $u$ is bordered.
\end{proof}
\begin{lemma}\label{lemma:shortestHalf}
    Let $w$ be a length-$n$ bordered word. Let $u$ be the shortest border of $w$. Then $|u| \leq n/2$.
\end{lemma}
\begin{proof}
Suppose to the contrary, that $|u|> n/2$. Then $u$ must overlap itself within $w$. So $u$ is bordered. But this border of $u$ must also be a border of $w$. This contradicts the assumption that $u$ is the shortest border of $w$. Thus $|u| \leq n/2$. 
\end{proof}

\section{Recurrence}\label{section:recurrence}
In this section we give a recurrence for the number of bordered words with a given prefix. This is the basis for our ranking algorithm. Let $n\geq 1$ be an integer. Let $u$ be a word of length less than or equal to $n$. Let $B_k(u,n)$ denote the number of length-$n$ bordered words over $\Sigma_k$ that have $u$ as a prefix.
\begin{theorem}
    Let $n\geq p \geq 1$ and $k\geq 2$ be integers. Let $u$ be a length-$p$ word over $\Sigma_k$. Let $b[1..p]$ be the border indicator of $u$. Let $a[1..p]$ be the unbordered prefix indicator of $u$. Then
    \[B_k(u,n) = \begin{cases} 
      \sum\limits_{i=1}^{n-p}a[i]k^{n-p-i} + \sum\limits_{i=n-p+1}^{\lfloor n/2\rfloor} a[i]b[i-(n-p)], & \text{if $n \leq 2p$;} \\
      \sum\limits_{i=1}^p a[i] k^{n-p-i} + \sum\limits_{i=p+1}^{\lfloor n/2\rfloor} (k^{i-p} - B_k(u,i))k^{n-2i}, & \text{otherwise.}
   \end{cases}\]

\end{theorem}
\begin{proof}
    Let $w$ be a length-$n$ bordered word. Suppose that $u$ is a prefix of $w$. Let $v$ be the shortest border of $w$. Our strategy is to split up the set of all length-$n$ bordered words with $u$ as a prefix into sets, $S_1, \ldots, S_n$, such that $\sum_{i=1}^n|S_i| = B_k(u,n)$. We choose $S_i$ to be the set of all length-$n$ bordered words that have $u$ as a prefix and have a length-$i$ shortest border. These sets are clearly disjoint, and their union is just the set of all length-$n$ bordered words with $u$ as a prefix. By Lemma~\ref{lemma:shortestUnb} and Lemma~\ref{lemma:shortestHalf} we have that $v$ is unbordered and $|v| \leq n/2$. Therefore, we have $|S_i|=0$ for $i> n/2$. Since $v$ is unbordered, we have $a[|v|] = 1$ when $|v| \leq |u|=p$. There are two cases to consider. The first case is when $|u|=p$ is greater than or equal to half the length of $w$, or $2p\geq n$. The other case is when $2p < n$.
    
    Suppose $2p \geq n$. Since $v$ is unbordered and $|v|\leq n/2 \leq p$, we have that $v$ must occur as an unbordered prefix of $u$. We further split this case into two subcases, one where $|v| \leq n-p$ and one where $n-p+1\leq |v| \leq n/2$. If $|v| \leq n-p$, then $w = vxyv$ where $u = vx$ and $y$ is a word of length $n-p-|v|$. Since $v$ and $x$ are determined by $u$, we have that there are $k^{n-p-|v|}$ choices for $w$. If $n-p+1\leq |v| \leq n/2$, then the instance of $v$ that is a suffix of $w$ must also have a non-empty overlap with $u$. So $w = vxyz$ where $x$ is a possibly empty word and $y,z$ are non-empty words such that $u=vxy$ and $v=yz$. But this means that $u=yzxy$. Thus $u$ must have a border of length $|y| = |v|-|z| = |v| - (n-p)$. So $w$ is bordered if and only if $b[{|v|-(n-p)}]=1$. When summing over all possible prefixes $v$ for the case when $2p\geq n$, we use the unbordered indicator $a[1..p]$ to only include those $v$ that are unbordered. So we get that 
    \[ B_k(u,n) =\sum\limits_{i=1}^{n-p}a[i]k^{n-p-i} + \sum\limits_{i=n-p+1}^{\lfloor n/2\rfloor} a[i]b[{i-(n-p)}].\]

    Suppose $2p < n$. Again, we further split this case into two subcases, one where $|v| \leq p$, and one where $p+1\leq |v|\leq n/2$. If $|v| \leq p$, then $v$ must occur as an unbordered prefix of $u$. As in the previous case, this leads to there being $k^{n-p-|v|}$ choices for $w$. So suppose $p+1\leq |v|\leq n/2$. In this case, we have that $u$ must occur as a prefix of $v$. So this means that $v$ is unbordered and has $u$ as a prefix. The number of such words $v$ is clearly $k^{|v|-p}- B_k(u,|v|)$ (i.e., all length-$|v|$ bordered words with $u$ as a prefix subtracted from all length-$|v|$ words with $u$ as a prefix). We can write $w = vxv$ where $x$ is a word of length $n-2|v|$. There are $k^{|v|-p}- B_k(u,|v|)$ choices for $v$ and $k^{n-2|v|}$ choices for $x$. Therefore, there are $(k^{|v|-p}- B_k(u,|v|))k^{n-2|v|}$ choices for $w$ in this case. Summing over all prefixes $v$ for the case $2p < n$, and using the unbordered indicator $a[1..p]$ to only include those $v$ that are unbordered and are of length less than or equal to $|u|=p$, we get 
    \[B_k(u,n) =\sum\limits_{i=1}^p a[i] k^{n-p-i} + \sum\limits_{i=p+1}^{\lfloor n/2\rfloor} (k^{i-p} - B_k(u,i))k^{n-2i}.\]
\end{proof}
\begin{theorem}\label{theorem:calcB}
Let $n\geq 1$. Let $u$ be a word of length less than or equal to $n$. The recurrence $B_k(u,n)$ can be computed in $O(n^2)$ time using $O(n)$ space.
\end{theorem}
\begin{proof}
Before starting the computation of $B_k(u,n)$, we first need to compute the border indicator of $u$ and the unbordered prefix indicator of $u$. From Lemma~\ref{lemma:unbPref} and Lemma~\ref{lemma:bordInd} we see that both the border indicator and unbordered prefix indicator of $u$ can be computed in $O(n)\in O(n^2)$ time. For each $i\geq |u|$, we have that calculating $B_k(u,i)$ involves computing a sum of $O(i)$ terms. Thus, using standard dynamic programming techniques, we can compute $B_k(u,n)$ in $O(n^2)$ time using $O(n)$ space. 
\end{proof}

\section{Ranking}\label{section:ranking}
In this section we show how to efficiently calculate the ranks of bordered and unbordered words. Let $\rankB_k(w)$ (resp. $\rankU_k(w)$) denote the rank of the word $w$ in the lexicographic listing of bordered (resp. unbordered) words of length $|w|$ over the alphabet $\Sigma_k$. 
\begin{theorem}\label{theorem:rankB}
Let $n\geq 1$ and $k\geq 2$ be integers. Let $w=w_1w_2\cdots w_n$ be a length-$n$ bordered word over $\Sigma_k$. Then
\[\rankB_k(w) = 1+\sum_{i=1}^{n}\sum_{c=1}^{w_{i}-1}B_k(w_1w_2\cdots w_{i-1} c, n).\]
\end{theorem}
\begin{proof} 
The rank of $w$ is just the position of $w$ in the lexicographic listing of all length-$n$ bordered words. This is equal to $1$ plus the number of length-$n$ bordered words that are smaller than $w$ in lexicographic order. By definition, for any word $u=u_1u_2\cdots u_n$ that is smaller than $w$, there exists an $i$ such that $u_1u_2\cdots u_{i-1}=w_1w_2\cdots w_{i-1}$ and $u_i<w_i$. Thus, any length-$n$ bordered word that begins with $w_1w_2\cdots w_{i-1}c$ for some $c< w_i$ is smaller than $w$. Summing over all possible $i$ and $c<w_i$ we have that the number of length-$n$ bordered words that are smaller than $w$ is \[\sum_{i=1}^{n}\sum_{c=1}^{w_{i}-1}B_k(w_1w_2\cdots w_{i-1} c, n).\] 
\end{proof}

\begin{theorem}\label{theorem:rankU}
Let $n\geq 1$ and $k\geq 2$ be integers. Let $w=w_1w_2\cdots w_n$ be a length-$n$ unbordered word over $\Sigma_k$. Then
\[\rankU_k(w) = 2+\sum_{i=1}^{n}(w_{i}-1)k^{n-i}-\rankB_k(w).\]
\end{theorem}
\begin{proof}
This proof follows the same structure as the proof of Theorem~\ref{theorem:rankB}. The rank of $w$ is equal to $1$ plus the number of length-$n$ unbordered words that are smaller than $w$. By definition, a length-$n$ word $u=u_1u_2\cdots u_n$ is smaller than $w$ if a there exists an $i$ such that $u_1u_2\cdots u_{i-1}=w_1w_2\cdots w_{i-1}$ and $u_i < w_i$. Therefore, any length-$n$ unbordered word that begins with $w_1w_2\cdots w_{i-1}c$ for $c<w_i$ is smaller than $w$. The number of length-$n$ unbordered words that begin with $w_1w_2\cdots w_{i-1}c$ is equal to the number length-$n$ bordered words that begin with $w_1w_2\cdots w_{i-1}c$ subtracted from all length-$n$ words that begin with $w_1w_2\cdots w_{i-1}c$. This is just equal to $k^{n-i} - B_k(w_1w_2\cdots w_{i-1}c, n)$. Summing over all possible $i$ and $c< w_i$, we have that the number of length-$n$ unbordered words that are smaller than $w$ is
\begin{align}
\sum_{i=1}^{n}\sum_{c=1}^{w_{i}-1}(k^{n-i}-B_k(w_1w_2\cdots w_{i-1} c, n)) &=\sum_{i=1}^{n}\sum_{c=1}^{w_{i}-1}k^{n-i}-\sum_{i=1}^{n}\sum_{c=1}^{w_{i}-1}B_k(w_1w_2\cdots w_{i-1} c, n)\nonumber \\
&= 1 + \sum_{i=1}^n(w_i-1)k^{n-i} - \rankB_k(w).\nonumber\end{align}
\end{proof}

\begin{algorithm}[H] 
\small
\caption{Computing $\rankB_k(w)$ and  $\rankU_k(w)$ given $w$ and $k$.}\label{algo:computeRankB}
\begin{algorithmic}[1]

\Statex
\Function{rankB}{$w=w_1w_2\cdots w_n,k$}
    \State{$total\gets 0$}
    
    \For{$i\gets 1$ to $n$}
        \For{$c\gets 1$ to $w_{i}-1$}
            \State $total \gets total + B_k(w_1w_2\cdots w_{i-1}c,n)$
        \EndFor
    \EndFor
    \State \Return $1+total$
\EndFunction

\Statex
\Function{rankU}{$w=w_1w_2\cdots w_n,k$}

    \State{$total\gets 0$}
    
    \For{$i\gets 1$ to $n$}
        \State $total \gets total + (w_{i}-1)k^{n-i}$
    \EndFor
    \State \Return $2+total - \Call{rankB}{w}$

\EndFunction

\end{algorithmic}
\end{algorithm}

Suppose $w$ is of length $n$. While computing $\rankB_k(w)$ in Algorithm~\ref{algo:computeRankB}, we have that for every index $i$ of $w$, we loop through $w_i-1\in O(k)$ different symbols. For each of these symbols we compute $B_k(u,n)$ for some word $u$ of length $i$, which takes at most $O(n^2)$ time using $O(n)$ space by Theorem~\ref{theorem:calcB}. Thus, we can compute $\rankB_k(w)$ in $O(kn^3)$. Since we are only required to store the length-$n$ word $w$, an integer variable which stores the rank, and a single length-$n$ integer array that we reuse to calculate $B_k(u,n)$, we have that computing $\rankB_k(w)$ uses only $O(n)$ space. Since $\rankU_k(w)$ uses $\rankB_k(w)$ and does only $O(n)$ extra work, it also runs in $O(kn^3)$ time using $O(n)$ space.

\begin{theorem}
    Let $n\geq 1$ and $k\geq 2$ be integers. Let $w$ be a length-$n$ word over $\Sigma_k$. Then $\rankB_k(w)$ can be computed in $O(kn^3)$ time using $O(n)$ space.
\end{theorem}
\begin{theorem}
    Let $n\geq 1$ and $k\geq 2$ be integers. Let $w$ be a length-$n$ word over $\Sigma_k$. Then $\rankU_k(w)$ can be computed in $O(kn^3)$ time using $O(n)$ space.
\end{theorem}

\section{Unranking}\label{section:unrank}
In this section we show how to unrank bordered and unbordered words. We describe the algorithm to compute the length-$n$ bordered word with rank $r$. The same algorithm can be adapted to unrank unbordered words as well.

Suppose $w=w_1w_2\cdots w_n$ is a length-$n$ bordered word with rank $r$ that we are trying to determine. First, we initialize $w$ to be $1^n$, the lexicographically smallest length-$n$ word. Then, starting with $i=1$, we determine $w_i$ by applying binary search on $1,2,\ldots, k$ to find the largest $x$ such that the rank of $w_1w_2\cdots w_{i-1}x 1^{n-j}$ is less than or equal to $r$. Then we set $w_i=x$, and move on to $w_{i+1}$. 

The reason we are choosing the largest $x$ is due to the fact that $\rankB_k$ can be applied to both bordered and unbordered words. It is easy to show that the largest word with a given rank is the desired bordered word at that rank. Observe that when unranking unbordered words, choosing the largest word at a particular rank can output a bordered word if the rank is larger than the number of unbordered words. In this case, the output is $kk\cdots k$.



\begin{algorithm}[H] 
\small
\caption{Unranking bordered words given a rank $r$, a length $n$, alphabet size $k$.}\label{algo:computeUnrankB}
\begin{algorithmic}[1]

\Statex
\Function{unrankB}{$r, n, k$}
    \State{$w_1w_2\cdots w_n\gets 11\cdots 1$}
    \For{$i\gets 1$ to $n$}
        \State $left \gets 1$
        \State $right \gets k$
        \While{$left < right$}
            \State $save \gets w_i$
            \State $mid \gets \lceil(left+right)/2\rceil$
            \State $w_i \gets mid$
            \If{$\Call{rankB}{w_1w_2\cdots w_n,k} \leq r$} $left \gets mid$
            \Else
                \State $w_i\gets save$
                \State $right \gets mid$
            \EndIf
            
        \EndWhile
    \EndFor
    \State \Return $w_1w_2\cdots w_n$
\EndFunction

\end{algorithmic}
\end{algorithm}
For every index $i$ of $w$ we perform a binary search on the alphabet of size $k$. For each iteration in the binary search we calculate $\rankB_k(w)$. There are $n$ indices of $w$, binary search on the alphabet of size $k$ takes $O(\log k)$ time, and calculating $\rankB_k(w)$ takes $O(kn^3)$ time. Thus, running $\Call{UnRankB}{r,n,k}$ takes $O(n^4k\log k)$ time. Running $\Call{UnRankB}{r,n,k}$ also takes $O(n)$ space since it only requires us to use a length-$n$ word, and a constant number of integer variables. Additionally, when calculating $\rankB_k(w)$ within the unranking procedure we can reuse a length-$n$ integer array to store $B_k(u,n)$. To adapt the algorithm to unrank unbordered words, just replace $\Call{rankB}{w,k}$ with $\Call{rankU}{w,k}$.

\begin{theorem}
Let $n\geq 1$ and $k\geq 2$ be integers. A bordered word $w\in\Sigma_k^n$ with rank $r$ can be computed in $O(n^4k\log k)$ time using $O(n)$ space.
\end{theorem}

\begin{theorem}
 Let $n\geq 1$ and $k\geq 2$ be integers. An unbordered word $w\in\Sigma_k^n$ with rank $r$ can be computed in $O(n^4k\log k)$ time using $O(n)$ space.
\end{theorem}

\section{Conclusions and Open Problems}\label{section:conclusion}
In this paper we presented the first efficient ranking and unranking algorithms for bordered and unbordered words. We gave a $O(kn^3)$ ranking algorithm that uses only $O(n)$ space for bordered and unbordered words. We also showed how to use this ranking algorithm to unrank bordered and unbordered words in $O(n^4k\log k)$ time using $O(n)$ space. 
We conclude by posing an open problem.
\begin{itemize}
    \item Can ranking and unranking bordered and unbordered words be done more efficiently? Can they be done in quadratic time?  
    \begin{itemize}
        \item It seems believable that the ranking and unranking algorithms in this paper can be improved by a factor of $n$. We briefly lay out the reasoning. As is, the recurrence $B_k(u,n)$ takes $O(n^2)$ time to determine. However, if we compute $B_k(u,n) - kB_k(u,n-1)$, all but a constant number of terms cancel out in the case that $n> 2|u|$. In the case that $n\leq 2|u|$, we are still left with $O(n)$ terms in the summation. If this summation could be computed in constant time, then computing $B_k(u,n)$ would take $O(n)$ time, effectively reducing the time complexity of ranking and unranking by a factor of $n$.
    \end{itemize}
\end{itemize}
\bibliographystyle{unsrt}
\bibliography{abbrevs,main}




\end{document}